\documentclass[a4paper,twocolumn,11pt,pra,accepted=2025-01-23]{quantumarticle}
\pdfoutput=1
\usepackage[T1]{fontenc} 
\usepackage[english]{babel} 
\usepackage{amsmath,amsfonts,amsthm} 
\usepackage{lipsum} 

\usepackage{mathtools}
\usepackage{dsfont}
\usepackage{graphics}
\usepackage{amssymb} 
\usepackage{graphicx}
\usepackage{color}
\usepackage{xcolor}
\usepackage{hyperref}
\usepackage{bm}
\usepackage{caption}
\usepackage{subcaption} 

\usepackage[numbers]{natbib}
\setcitestyle{shortnames=et~al}

\newcommand{\blockcomment}[1]{} 

\DeclareMathOperator{\tr}{Tr}

\newcommand{\id}{\mathds{1}}

\newcommand{\bra}[1]{\left\langle #1 \right|}
\newcommand{\ket}[1]{\left| #1 \right\rangle}
\newcommand{\braket}[2]{\left\langle #1 \middle| #2 \right\rangle}
\newcommand{\ketbra}[2]{\left|#1\middle\rangle\middle\langle#2\right|}
\newcommand{\proj}[1]{\left|#1\middle\rangle\middle\langle#1\right|}

\newtheorem{theorem}{Theorem}

\newtheorem{corollary}[theorem]{Corollary}

\newtheorem{lemma}[theorem]{Lemma}
\newtheorem*{lemma*}{Lemma}
\newtheorem{definition}[theorem]{Definition}

\makeatletter
\newcommand{\subalign}[1]{%
  \vcenter{%
    \Let@ \restore@math@cr \default@tag
    \baselineskip\fontdimen10 \scriptfont\tw@
    \advance\baselineskip\fontdimen12 \scriptfont\tw@
    \lineskip\thr@@\fontdimen8 \scriptfont\thr@@
    \lineskiplimit\lineskip
    \ialign{\hfil$\m@th\scriptstyle##$&$\m@th\scriptstyle{}##$\hfil\crcr
      #1\crcr
    }%
  }%
}
\makeatother

\makeatletter
\let\save@mathaccent\mathaccent
\newcommand*\if@single[3]{%
  \setbox0\hbox{${\mathaccent"0362{#1}}^H$}%
  \setbox2\hbox{${\mathaccent"0362{\kern0pt#1}}^H$}%
  \ifdim\ht0=\ht2 #3\else #2\fi
  }
\newcommand*\rel@kern[1]{\kern#1\dimexpr\macc@kerna}
\newcommand*\widebar[1]{\@ifnextchar^{{\wide@bar{#1}{0}}}{\wide@bar{#1}{1}}}
\newcommand*\wide@bar[2]{\if@single{#1}{\wide@bar@{#1}{#2}{1}}{\wide@bar@{#1}{#2}{2}}}
\newcommand*\wide@bar@[3]{%
  \begingroup
  \def\mathaccent##1##2{%
    \let\mathaccent\save@mathaccent
    \if#32 \let\macc@nucleus\first@char \fi
    \setbox\z@\hbox{$\macc@style{\macc@nucleus}_{}$}%
    \setbox\tw@\hbox{$\macc@style{\macc@nucleus}{}_{}$}%
    \dimen@\wd\tw@
    \advance\dimen@-\wd\z@
    \divide\dimen@ 3
    \@tempdima\wd\tw@
    \advance\@tempdima-\scriptspace
    \divide\@tempdima 10
    \advance\dimen@-\@tempdima
    \ifdim\dimen@>\z@ \dimen@0pt\fi
    \rel@kern{0.6}\kern-\dimen@
    \if#31
      \overline{\rel@kern{-0.6}\kern\dimen@\macc@nucleus\rel@kern{0.4}\kern\dimen@}%
      \advance\dimen@0.4\dimexpr\macc@kerna
      \let\final@kern#2%
      \ifdim\dimen@<\z@ \let\final@kern1\fi
      \if\final@kern1 \kern-\dimen@\fi
    \else
      \overline{\rel@kern{-0.6}\kern\dimen@#1}%
    \fi
  }%
  \macc@depth\@ne
  \let\math@bgroup\@empty \let\math@egroup\macc@set@skewchar
  \mathsurround\z@ \frozen@everymath{\mathgroup\macc@group\relax}%
  \macc@set@skewchar\relax
  \let\mathaccentV\macc@nested@a
  \if#31
    \macc@nested@a\relax111{#1}%
  \else
    \def\gobble@till@marker##1\endmarker{}%
    \futurelet\first@char\gobble@till@marker#1\endmarker
    \ifcat\noexpand\first@char A\else
      \def\first@char{}%
    \fi
    \macc@nested@a\relax111{\first@char}%
  \fi
  \endgroup
}
\makeatother

\newcommand{\T}{{\mathrm{ST}}}
\newcommand{\n}{{(n)}}
\newcommand{\on}{{\otimes n}}

\begin{document}

\title{A de Finetti theorem for quantum causal structures}
\author{F.~Costa}
\email{fabio.costa@su.se}
\affiliation{Nordita, Stockholm University and KTH Royal Institute of Technology, Hannes Alfv{\'e}ns v{\"a}g 12 Stockholm, 106 91, Sweden}
\affiliation{School of Mathematics and Physics, The University of Queensland, St Lucia, QLD 4072, Australia}
\author{J.~Barrett}
\affiliation{Quantum Group, Department of Computer Science, University of Oxford}
\author{S.~Shrapnel}
\affiliation{ARC Centre for Engineered Quantum Systems, School of Mathematics and Physics, The University of Queensland, St Lucia, QLD 4072, Australia}

\begin{abstract}
What does it mean for a causal structure to be `unknown'? Can we even talk about `repetitions' of an experiment without prior knowledge of causal relations? And under what conditions can we say that a set of processes with arbitrary, possibly indefinite, causal structure are independent and identically distributed? Similar questions for classical probabilities, quantum states, and quantum channels are beautifully answered by so-called ``de Finetti theorems'', which connect a simple and easy-to-justify condition---symmetry under exchange---with a very particular multipartite structure: a mixture of identical states/channels. 
Here we extend the result to processes with arbitrary causal structure, including indefinite causal order and multi-time, non-Markovian processes applicable to noisy quantum devices. The result also implies a new class of de Finetti theorems for quantum states subject to a large class of linear constraints, which can be of independent interest.
\end{abstract}

\maketitle

\section{Introduction}

The possibility to repeat an experiment is one of the fundamental tenets of the scientific method: after a sufficient number of repetitions, statistical analysis typically enables us to characterise the system or process of interest, to decide between competing hypotheses, and so on. However, a potentially unsettling question underlies this paradigm: what counts as a repetition? If we toss a coin hundreds of thousands of times \cite{Bartos2023}, are we repeating the same experiment or are we rather performing a set of different experiments? And are we allowed to combine tosses from different coins?

The de Finetti theorem \cite{DeFinetti1929, deFinetti1937, Hewitt1955} offers an elegant answer: it states that a joint probability for a set of random variables that is $a)$ invariant under permutations and $b)$ the marginal of a probability of an arbitrarily larger set, also invariant under permutations, must take the form
\begin{multline}
P^{(n)}(a^1,\dots,a^n) \\
 = \int dq P(q) q(a^1)\cdots q(a^n),
\label{classicaldefinetti}
\end{multline}
for a unique normalised measure $P(q)$ over the space of single-variable probabilities $q$. The striking consequence is that a set of exchangeable variables (i.e., satisfying conditions $a$ and $b$ above) can always be interpreted as independent and identically distributed (i.i.d.), up to a global uncertainty $P(q)$ on the single-trial probability $q$. 

De Finetti's theorem is particularly satisfying from a Bayesian point of view: one never needs to invoke an `unknown' probability---which would make little sense if probabilities represent degrees of belief---nor needs one introduce any ad-hoc notion of repeatability. As long as they can justify exchangeability, the Bayesian can treat a set of variables \emph{as if} they were multiple trials of the same experiment, with each trial distributed according to the `unknown' probability $q$, and where $P(q)$ represents the prior knowledge about such a probability. Conditioning on a larger and larger set of observed variables allows one to update $P(q)$, eventually converging to a particular single-trial $q^*$, $P(q) \rightarrow \delta(q - q^*)$. This justifies the idea of `discovering' the unknown probability through repeated trials and gives a principle-based support for Bayesian methods. 
Apart from the foundational significance, classical de Finetti-type theorems have extensive applications in pure and applied probability theory \cite{Kingman1978, Aldous1985, OBrien1988}.  

De Finetti's theorem extends in a direct way to quantum states \cite{Stoermer1969, Hudson1976} and channels \cite{Fuchs2004}: an exchangeable multipartite state/channel is always a mixture of product states/channels. Apart from liberating the Bayesian from the uncomfortable notion of `unknown state'\footnote{Regardless of one's ontological view on pure quantum states, it should always be possible to use mixed states to represent incomplete subjective  knowledge. In this sense, `unknown mixed states' are as problematic as `unknown probability distributions', calling for a principle-based approach to repeatability and discovery.} \cite{Caves2002a} and grounding the use of Bayesian methods in quantum information \cite{Schack2001, Granade2017qinferstatistical}, these results have wide-ranging applications in many-body physics \cite{Fannes1980, Krumnow2017}, cryptography \cite{Renner2007, Christandl2009, Renner2009}, quantum information \cite{Brandao2010, Navascues2009, Brandao2011a, Brandao2011, Brandao2017}, and quantum foundations \cite{Hudson1981, Barrett2009, Christandl2009a, ArnonFriedman2015}.

However, some natural questions emerging in quantum theory escape known de Finetti results. Consider an experiment that comprises multiple measurements and operations distributed in space and time. Can the causal relations between such operations be \emph{unknown} and discovered through multiple trials? Does it even make sense to \emph{repeat} such an experiment, given that each spacetime event can only happen once?

Within a broader effort to understand the role of causal structure in quantum theory \cite{Laskey2007, Leifer2013, cavalcanti2014modifications,fritzbeyond2015, woodlesson2012, Henson2015, pienaar2014graph, chavesinformation2015, ried2015quantum, costa2016, Shrapnel2018causation, Allen2016, Giarmatzi2018, Barrett2019, Pearl2021}, a framework has recently emerged in which causal relations need not be fixed in advance \cite{oreshkov12, oreshkov15, araujo15} and can be discovered through experiments. The framework---often dubbed the ``process matrix formalism''---also includes scenarios where causal relations are genuinely indefinite \cite{chiribella09}, with potential applications to quantum information processing \cite{araujo14, feixquantum2015, Guerin2016, Jia2019, Goswami2020} and fundamental models of quantum gravity \cite{hardy2007towards, Zych2019, Hardy2020, parker2021background}. Furthermore, the special case of causally ordered, multi-time processes is emerging as a powerful tool to tackle temporal correlations in non-Markovian quantum processes \cite{modioperational2012, Milz2017, Pollock2018, Shrapnel2018, giarmatzi2018witnessing, Luchnikov2019}, an increasingly prominent feature in complex quantum devices \cite{Morris2019, Sandia_report}. 

\begin{figure}%
\begin{center}
\includegraphics[width=\columnwidth]{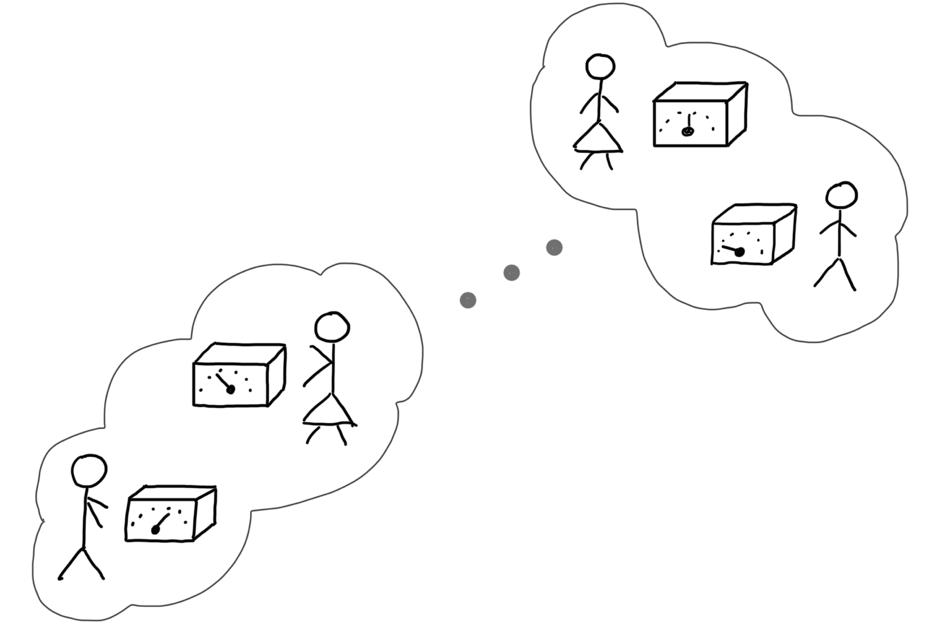}%
\end{center}
\caption{\textbf{Process repeatability}. Repetitions of an experiment are modelled as a one-shot scenario, comprising multiple operations that are only performed once. A de Finetti theorem for processes seeks to group these operations such that they represent independent trials under equivalent conditions, with each trial involving multiple operations in arbitrary, possibly unknown or indefinite, causal relations.}%
\label{processrepetitions}%
\end{figure}

As for ordinary quantum theory, the process matrix formalism makes probabilistic predictions, tacitly assuming that the processes it describes can be repeated arbitrarily many times. Indeed, the paradigm has been employed in several experiments probing processes with definite \cite{White2020, Goswami2021, White2022, Xiang2021, giarmatzi2023multitime} and indefinite \cite{Procopio2015, Rubinoe1602589, rubino2017experimental, Goswami2018, guo2018experimental, goswami2018communicating, Wei2019, taddei2020experimental} causal structure. However, a foundational justification for repeatability is missing. In fact, without additional assumptions, even the quantum de Finetti theorem for states does not apply to typical `repeated state preparation' scenarios: If the repetitions are temporally separated, they cannot be modelled as a joint state, as this would constitute a ``state over time'' \cite{Horsman2017}. This undermines the applicability of standard statistical analysis and characterisation techniques in the most common quantum experiments. 

The goal of this work is to extend the link between exchangeability and repeatability to quantum processes with arbitrary causal structure. To this end, multiple putative trials of an experiment are modelled as a single process, where all measurements and operations are performed only once, Fig.~\ref{processrepetitions}. A global process matrix encapsulates prior knowledge about the entire setup, and the task is to establish that an assumption of exchangeability enables the recovery of a decomposition into several repetitions of `the same' process. This requires extending the de Finetti theorem to quantum processes with general causal structures. At first glance, it may seem that any additional causal constraint (such as global causal order or other no-signalling assumptions) necessitates its own de Finetti theorem. This is because, as made clear through the process-state duality, each such assumption is equivalent to a different set of linear constraints on quantum states, which, a priori, need not be preserved by the de Finetti representation. 

Here we solve the above hurdle by proving a new generalised de Finetti theorem for exchangeable states subject to a broad class of linear constraints. We then specialise the general theorem to show that an exchangeable process subject to a set of no-signalling constraints has a unique representation as a mixture of i.i.d.\ processes subject to the same constraints. We further consider possible extensions of the result and find that, remarkably, even mild generalisations of the theorem's hypothesis do not lead to corresponding de Finetti representations. 

\section{The process matrix formalism}

The process matrix formalism \cite{oreshkov12, oreshkov15, araujo15} characterises the most general scenario where a quantum system, or multiple quantum systems, can be probed an arbitrary number of times, without prior assumptions regarding spatiotemporal or causal relations between different operations. It relates closely to other frameworks, such as the general boundary formalism \cite{Oeckl2003318}, higher-order quantum transformations \cite{chiribellatransforming2008, Perinotti2017,Bisio2019}, multi-time states \cite{Aharonov2009, Silva2014, Silva2017}, entangled histories \cite{Cotler2016} , and superdensity operators \cite{Cotler2017}.

 We call a \emph{site} the abstract location of an operation and we label sites as $A, B, \dots$ Concretely, these labels can be understood as spacetime coordinates at which the operations take place, or more general ways to identify the abstract location of the site, which may be delocalised in space and time.
The most general operation performed at a site $A$, yielding some measurement outcome $a$, is a completely positive (CP) and trace non-increasing map \cite{Heinosaari2011} $\mathcal{M}_a:\mathfrak{L}(\mathcal{H}^{A_I})\rightarrow \mathfrak{L}(\mathcal{H}^{A_O})$, where $\mathcal{H}^{A_I}$, $\mathcal{H}^{A_O}$ are respectively the input and output Hilbert spaces assigned to site $A$ and $\mathfrak{L}(\mathcal{H})$ denotes the set of linear operators on $\mathcal{H}$. We will always use (a version of) the Choi-Jamio{\l}kowski (CJ) isomorphism to represent CP maps \cite{jamio72,Choi1975}: $\mathcal{M}_a\mapsto M_a \in 
\mathfrak{L}(\mathcal{H}^{A})$,
\begin{equation}\label{Choi}
M_a\coloneqq \left[\sum_{jk} \ketbra{j}{k}\otimes \mathcal{M}_a(\ketbra{j}{k}) \right]^T,
\end{equation}
for a chosen orthonormal basis $\{\ket{j}\}_j$ of $\mathcal{H}^{A_I}$, where $^T$ denotes transposition in that basis and we use the short-hand $\mathcal{H}^{A} \equiv \mathcal{H}^{A_I}\otimes \mathcal{H}^{A_O}$. The CP condition translates to the CJ operator being positive semidefinite, $M^{A}_a\geq 0$. For simplicity, we will nominally identify CP maps with their CJ representations.

\begin{figure}[t!]%
\begin{center}
       (a) \includegraphics[width=0.4\columnwidth]{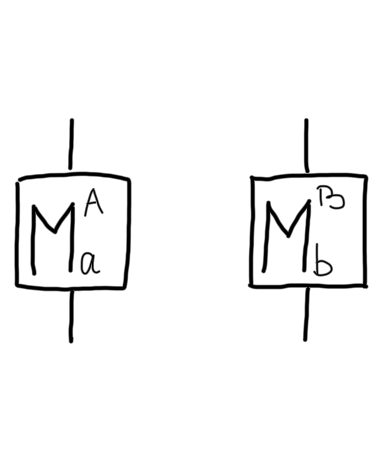} \quad (b) \includegraphics[width=0.38\columnwidth]{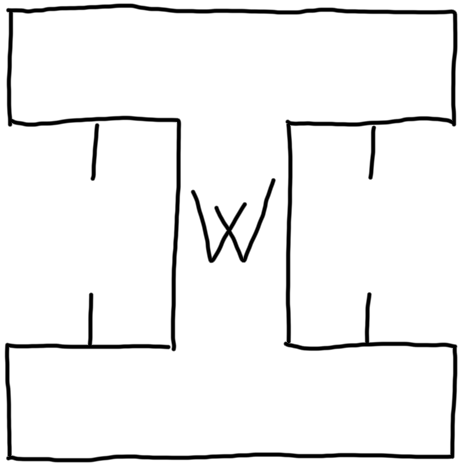} \\
			(c) \includegraphics[width=0.45\columnwidth]{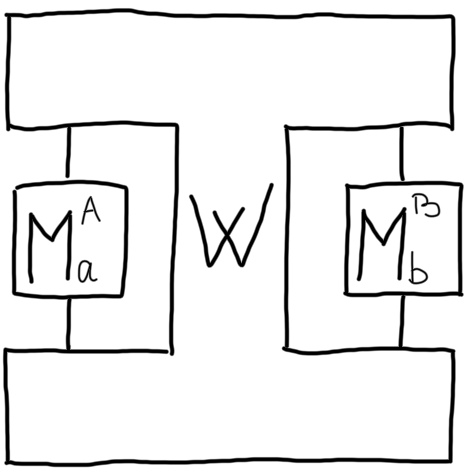}
\caption{\textbf{Operations and processes} (\textbf{a}) Quantum operations, represented as CJ operators $M^A_a$, $M^B_b$, transform an input to an output quantum system---here depicted as wires. The labels $A$, $B$ function as generalised coordinates, identifying the operations without necessarily referring to a background causal structure, while $a$, $b$ denote measurement outcomes. (\textbf{b}) A process matrix $W$ represents the most general way to connect operations. (\textbf{c}) Inserting operations into a process, with no open wires left, returns the probability for observing outcomes $a,b$ through the Born rule for processes, Eq.~\eqref{processborn_Npartite}. 
\label{processmatrix}}%
\end{center}%
\end{figure}

A deterministic operation $\widebar{M}$---one with a single measurement outcome that happens with probability one---is represented by a CP and trace preserving (CPTP) map, which in CJ form translates to the condition 
\begin{equation}
\label{CPTP}
\tr_{A_O}\widebar{M} = \id^{A_I},
\end{equation}
where $\id^X$ denotes the identity operator on $\mathcal{H}^X$ (we may skip the superscripts if the context is sufficiently clear). An \emph{instrument} is a collection of CP maps, $\left\{M_a\right\}_a$, $M_a\geq 0$, that sums to a CPTP map, $\tr_{A_O} \sum_a M_a= \id^{A_I}.$   
 Note that, even when an operation can be regarded as a transformation of a single system, we treat input and output as distinct (although possibly isomorphic) spaces. Furthermore, we assign different Hilbert spaces to different sites, even though they might represent the same system at different times. The space of operations across all sites spans the tensor product of all input and output spaces.

The probability to obtain a set of outcomes $a,b,\dots$ at sites $A,B, \dots$  
is given by a generalisation of the Born rule:
\begin{equation}\label{processborn_Npartite}
    P(a, b, \ldots) = \tr \left[\left(M^{A}_{a}\otimes M^{B}_{b}\cdots\right) W\right],
\end{equation} 
where the operator $W \in \mathfrak{L}(\mathcal{H}^{A}\otimes\mathcal{H}^{B}\otimes \cdots)$ is the \emph{process matrix}, Fig.~\ref{processmatrix},
 which satisfies 
\begin{align} \label{positivity}
W & \geq 0 , \\ \label{normalisation}
\tr & \left[\left(\widebar{M}^{A}\otimes \widebar{M}^{B}\cdots\right) W \right] = 1
\end{align}
for all CPTP maps $\widebar{M}^{A}, \widebar{M}^{B}, \dots$
These constraints can be derived on abstract grounds by assuming that quantum theory is valid at each site---which can be formalised, e.g., through a non-contextuality assumption \cite{Shrapnel2017}---and imposing positivity and normalisation of probabilities. They hold in particular for all ordinary scenarios in quantum mechanics, where all sites can be ordered according to a background time.

The normalisation constraint \eqref{normalisation} is equivalent to a set of linear-affine constraints:
\begin{align} \label{affine}
\tr W &= d^O, \\
L (W) &= 0,  
\end{align}
where $d^O = d^{A_O}d^{B_O}\dots$ is the product of all output-space dimensions and $L{\,:\,}{\mathfrak{L}(\mathcal{H}^{A}\otimes\mathcal{H}^{B}\otimes \cdots)}{\rightarrow}{\mathfrak{L}(\mathcal{H}^{A}\otimes\mathcal{H}^{B}\otimes \cdots)}\hspace{-0.4cm}$ is  a linear function whose particular form is not relevant here, see, e.g., Appendix B in Ref.~\cite{araujo15}.

Additional causal assumptions between sites---specifically, no signalling assumptions---can be enforced by adding linear constraints to the process matrix. For example, imposing no signalling between the sites $A$, $B$ of a bipartite process is equivalent to \cite{araujo15}
\begin{equation}
W = \left(Tr_{A_OB_O} W\right) \otimes \id^{A_OB_O},
\label{nosignalling}
\end{equation}
where (here and in similar expressions) a re-ordering of tensor factors is implied. 
Similarly, one-way no signalling from $B$ to $A$ corresponds to the linear constraints
\begin{align} \label{nosignallingAtoB1} 
W &= \left(Tr_{B_O} W \right) \otimes \id^{B_O}, \\ \label{nosignallingAtoB2} 
Tr_{B_IB_O} W &= \left(Tr_{A_OB_IB_O} W \right) \otimes \id^{A_O}.
\end{align}
Similar constraints characterise one-way no signalling for an arbitrary sequence of causally ordered sites $A, B, C,\dots$ Causally ordered processes, also known as quantum channels with memory \cite{Kretschmann2005}, quantum strategies \cite{gutoski06}, and quantum combs \cite{chiribella08, chiribella09b, Bisio2011}, can always be realised as a sequence of channels connecting the sites, possibly with the addition of an auxiliary system (an environment), thus modelling the most general non-Markovian open-system dynamics \cite{modioperational2012, Milz2017, Pollock2018}. More complex assumptions, such as an arbitrary partial-order relation between sites, can be obtained by combining non-signalling constraints \cite{costa2016, Giarmatzi2018}.

In the following, we will need the \emph{reduced process matrix} obtained by removing one or more sites from a larger set \cite{araujo15}. Unlike for states, the reduced process matrix is not always unique: if a site $B$ has causal influence on (i.e., it can signal to) a site $A$, then, by definition of signalling, the reduced process for $A$ can depend on the operation performed at $B$:
\begin{equation}
W^{A}_{\widebar{M}^{B}}:= \tr_B \left[\left(\id^{A}\otimes \widebar{M}^{B}\right) W^{AB}\right],
\label{reduced}
\end{equation}
where $\widebar{M}^B=\sum_b M^B_b$ is the CPTP map representing the unconditional transformation at $B$ if we ignore the outcomes of the instrument $\{M^B_b\}_b$.

\section{Exchangeable processes}
 
As in ordinary quantum mechanics, the process matrix formalism makes probabilistic predictions, through the Born rule in Eq.~\eqref{processborn_Npartite}. In order for such predictions to be experimentally meaningful, one typically assumes that an experiment involving a number of sites $A, B, \dots$ can be repeated an arbitrary number $n$ of times, where each trial should be modelled by the same process $W$. In a consistent probabilistic framework---and in particular, in a Bayesian perspective---it should be possible to combine all trials into a single experiment, and then specify the assumptions that entitle us to view the total experiment as repetitions of individual trials under equivalent conditions.

\begin{figure}[t!]%
\begin{center}
       (a) \includegraphics[width=0.37\columnwidth]{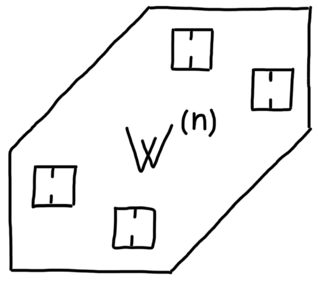} \quad (b) \includegraphics[width=0.41\columnwidth]{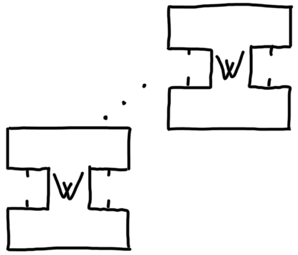} 
\caption{\textbf{Exchangeable and i.i.d.\ processes} (\textbf{a}) A priori, $n$ repetitions of an experiment constitute a one-shot process $W^\n$. (\textbf{b}) The process can be interpreted as $n$ i.i.d.\ trials if it is a product of identical processes, $W^\n=W^\on$. Our result guarantees that exchangeable processes are always mixtures of i.i.d.\ processes. \label{repeatedprocessmatrix}}%
\end{center}%
\end{figure}

With this in mind, we will use the labels $A, B, \dots$ to denote the sites in a single, unspecified trial, while we will use indices $A^1,\dots,A^n, B^1,\dots,B^n,\dots$ to distinguish the different trials. A priori, an $n$-trial scenario with $m$ sites per trial is described by a single one-shot scenario with $N=n\,m$ sites, with a process matrix 
$$W^\n \in \mathfrak{L}(\mathcal{H}_\T^\on),$$
Fig.~\ref{repeatedprocessmatrix}, where the single-trial Hilbert space $\mathcal{H}_\T$ is the tensor product of the input and output spaces associated with all $m$ single-trial sites: 
\begin{equation*}
\mathcal{H}_\T:= \mathcal{H}^{A}\otimes\mathcal{H}^{B}\otimes\cdots
\end{equation*}
In the following, the decomposition of an individual trial into sites $A, B, \dots$ will mostly be irrelevant, so we will effectively treat a single trial as a single site. In line with this, we introduce the short-hand notation $M^j\equiv M^{A^j}\otimes M^{B^j}\otimes\cdots$ to denote  a collection of CP maps for trial $j$, and simply refer to it as \emph{a} CP map. With this convention, conditions \eqref{positivity}, \eqref{normalisation} for a valid $n$-trial process matrix can be written as\footnote{Here and for most of this work, we depart from the common notation in the process matrix literature and drop the superscripts to denote tensor factors. Instead, unless otherwise specified, we assume that tensor factors are always in a reference order, with trials in increasing order from left to right.}
\begin{align} \label{npositivity}
W&^{(n)} \geq 0 , \\ \label{nnormalisation}
\tr & \left[\left(\widebar{M}^{1}\otimes \cdots \otimes \widebar{M}^{n} \right) W \right] = 1
\end{align}
for all CPTP maps $\widebar{M}^{1},  \cdots , \widebar{M}^{n}$. 


As in other de Finetti theorems, we need to capture the idea that different trials are equivalent to each other and that, in principle, there is no bound to the number of trials. The first condition is simply modelled by requiring invariance under relabelling of trials. More specifically, given an $n$-trial process matrix $W^{(n)}$, let us denote by $\ket{\mu}\in \mathcal{H}_\T$ the states in a chosen orthogonal basis of the single-trial space. We define the action of an $n$-element permutation $\sigma\in S_n$ as
\begin{equation}
U(\sigma):= \sum_{\mu_1\dots \mu_n} \ket{\mu_{1}\dots \mu_{n}}\bra{ \mu_{\sigma(1)}\dots \mu_{\sigma(n)}}.
\label{permutations}
\end{equation}
This means that all single-trial sites are permuted together: $A^{j}\mapsto A^{\sigma(j)}$, $B^{j}\mapsto B^{\sigma(j)}, \dots$ with the same permutation $\sigma$. 
The permuted process matrix is then 
\begin{equation}
W^{(n)}_{\sigma} := U(\sigma) W^{(n)} U^{\dagger}(\sigma),
\end{equation}
and we require that it should be equal to the initial process matrix.

The second condition, \emph{extendibility}, means that an $n$-trial scenario can be obtained by `ignoring' the last trial in an $n+1$-trial scenario: $\rho^{(n)}=\tr_{n+1}\rho^{(n+1)}$ for states. For process matrices, there are potentially different candidate definitions, because of the non-uniqueness of the reduced process mentioned above. A simple choice, which is sufficient to deduce a de Finetti representation, is to require that no signalling is possible from the sites in the $n+1$ trial to all the sites from trial $1$ to $n$. With this choice, we have the following definition:

\begin{definition}\label{processexchangeability}
A sequence of process matrices $W^{(n)}$, $n\geq 1$, is called \emph{exchangeable} if it satisfies
\begin{enumerate}
	\item \textbf{Symmetry:} $W^{(n)}$ is invariant under permutations of trials:
	\begin{equation}
	W^{(n)}_{\sigma} = W^{(n)} \qquad \forall\, \sigma\in S_n;
	\label{symmetric}
	\end{equation}
	\item \textbf{Process extendibility:} For every $n$ and for every CPTP map $\widebar{M}$,
	\begin{equation}
	\! W^{(n)} = \tr_{n+1}\left[W^{(n+1)}(\id^{\on} \otimes \widebar{M})\right],
	\label{pmextendibility}
	\end{equation}
	where the partial trace is over the $n+1$ factor in $\mathcal{H}_\T^{\otimes (n+1)}$.
\end{enumerate}
We say that an $n$-trial process matrix is \emph{exchangeable} if it is part of an exchangeable sequence.
\end{definition}
In the definition of extendibility, no-signalling is enforced by requiring that we obtain the same $W^{(n)}$ for any CPTP map $\widebar{M}$ in the last trial. Together with the symmetry assumption, this implies that no signalling is possible between any (sets of) trials. We explore the consequences of different extendibility conditions in Appendix \ref{otherextendibility}.

\section{Constrained states and processes}
Since process matrices are positive semidefinite and have fixed trace, they are always proportional to density operators, up to a constant. This implies that it is possible to treat processes and states in a unified way, which also allows us to leverage known results, such as the quantum de Finetti theorem for states. To this end, it is useful to work with re-normalised process matrices
\begin{equation}
\rho:=W/\tr W = W/d_O.
\label{processtostate}
\end{equation}

The process matrix normalisation constraints, Eq.~\eqref{nnormalisation}, are strictly stronger than state normalisation, $\tr\rho=1$. Therefore, the mapping \eqref{processtostate} identifies processes with states that are subject to additional linear constraints. 
In view of identifying a broader class of constraints, a natural generalisation of Eq.~\eqref{nnormalisation} is the following:
\begin{definition}
Given a set of single-trial operators $\mathfrak{R}\subset \mathfrak{L}\left(\mathcal{H}_\T\right)$, and a function $r:\mathfrak{R}\rightarrow \mathbb{C}$, we say that an $n$-trial operator $\rho^{(n)}$ satisfies a \emph{product expectation constraint} if
\begin{equation}
\tr\left[\left(\bigotimes_{j=1}^n  R_{j} \right)\rho^{(n)}\right] = \prod_{j=1}^n r(R_{j})
	\label{compatible}
	\end{equation}
for all choices of $R_{1},\dots,R_{n} \in \mathfrak{R}$.
\end{definition}

The process normalisation constraints, Eq.~\eqref{nnormalisation}, correspond to 
\begin{multline} \label{CPTPchoi}
\mathfrak{R} = \left\{R=R^A\otimes R^B\otimes\dots \in \mathfrak{L}\left(\mathcal{H}_\T\right) \right. \\
\textup{s.t. }\left. \tr_{X_O} R^X = d^{X_O} \id^{X_I},\; X=A,B,\dots  \right\}
\end{multline}
and $r(R)=1$, where the conditions in Eq.~\eqref{CPTPchoi} are the rescaled CJ form of the trace-preserving condition, Eq.~\eqref{CPTP}. More generally, constraints of the form \eqref{compatible} can be interpreted as the assumption that the expectation value of tensor products of certain observables is equal to the product of the expectation values. See Sec.~\ref{frequentist} below for a further discussion.
 
In order to simplify the proof of our main result, and make it directly applicable to a variety of scenarios, it is useful to re-write, and further generalise, constraints of the form \eqref{compatible}. We present here the key steps and leave proof details to Appendix \ref{lemmaproofs}.

The first observation is that the affine component of the constraints---i.e., the term on the right-hand side of Eq.~\eqref{compatible}---can be absorbed in the state constraint $\tr\rho^{(n)} = 1$ (which we always assume), while maintaining the product structure of the constraint. This can be stated as follows:
\begin{lemma}\label{unbiasedconstraintlemma}
A sequence of exchangeable states $\rho^{(n)}$ satisfies product expectation constraints of the form \eqref{compatible}, for a set of operator $\mathfrak{R}\subset \mathfrak{L}\left(\mathcal{H}_\T\right)$ and a function $r:\mathfrak{R}\rightarrow \mathbb{C}$, if and only
\begin{equation}
\tr[\left(\bigotimes_{j=1}^n  \sigma_{j} \right)  \rho^{(n)}] = 0
\label{linearcompatible}
\end{equation}
for all 
\begin{equation}
\sigma_j \coloneqq R_j - r(R_j), 
\label{Rtosigma}
\end{equation} 
$R_j\in\mathfrak{R}$, $j=1,\dots,n$.
\end{lemma} 

Next, we observe that a set of linear constraints, such as Eq.~\eqref{linearcompatible}, can always be rewritten as a single linear constraint defined by an appropriate vector-valued function. Explicitly, we prove in Appendix \ref{linearlemmaappendix} that   
\begin{lemma} \label{linearlemma}
Given two inner product spaces $V^1$, $V^2$, a measurable space $Y$, and a set of linear functions $L_y: V^1\rightarrow V^2$, $y\in Y$,
\begin{equation} \label{almostforall}
L_y(v)=0 \quad \text{\emph{a.e. for }} y\in Y
\end{equation}
for every $v\in V^1$ if and only if, for any strictly positive measure $q$ over $Y$ ($q(y)>0$ a.e.)
\begin{equation}
\int_Y dy \, q(y)L^{\dag}_y\circ L_y(v)=0,
\end{equation}
where ``a.e.'' stands for ``almost everywhere'', meaning that the property in question may not hold in at most a measure-zero set.
\end{lemma}

As an example of Lemma \ref{linearlemma} in action, consider a linear constraint as in Eq.~\eqref{linearcompatible}, for $n=1$. Defining $L_R(\rho) \coloneqq \tr \sigma \rho$, where the dependence on $R$ is given by Eq.~\eqref{Rtosigma}, the constraint is defined by a (possibly infinite) set of equations: $L_R(\rho)= 0$ $\forall R\in \mathfrak{R}$. However, thanks to Lemma \ref{linearlemma}, these constraints are equivalent to the single (operator valued) equation
\begin{multline}
\int_{\mathfrak{R}} dR \, L^{\dag}_R\circ L_R(\rho) \\
 = \int_{\mathfrak{R}} dR \; \sigma \tr \left[\sigma \rho\right] = 0.
\label{linearproductsigma}
\end{multline}

Using the same notation, the constraint in Eq.~\eqref{linearcompatible} for $n>1$ is given by a set of equations  $L_{R_1}\otimes\dots\otimes L_{R_n}(\rho^\n)= 0$, for all combinations of $R_1,\dots,R_n \in \mathfrak{R}$. In particular, this implies the constraints $(L_R)^\on(\rho^\n)=0$, with the same $R\in \mathfrak{R}$ for every trial. Using again Lemma \ref{linearlemma}, this can be expressed with a single equation:
\begin{equation}
 \int_{\mathfrak{R}} dR \; (L^{\dag}_R\circ L_R)^\on(\rho^\n) = 0 .
\end{equation}
Even though, for generic states, this is a strictly weaker condition than Eq.~\eqref{linearcompatible}, it will turn out to be sufficient to prove our constrained de Finetti theorem. This motivates us to introduce the following class of \emph{de Finetti-type} constraints:

\begin{definition}
Let $\mathfrak{H}(\mathcal{H})$ denote the real vector space of self-adjoint operators on a Hilbert space $\mathcal{H}$. Given a set of real vector spaces $\left\{V^k\right\}_k$, a measurable space $Y$, a set of linear functions $L^k_y:\mathfrak{H}(\mathcal{H}_\T)\rightarrow V^k$ defined for all $y\in Y$, and a set of strictly positive measures $q^k$ over $Y$, we say that an $n$-trial state $\rho^{(n)}$ satisfies \emph{de Finetti-type} constraints if
\begin{equation}
\int_Y dy \, q^k(y) \left(L^k_y \right)^{\otimes n}(\rho^{(n)}) =0 \quad \forall k.
\label{exchangeableconstraint}
\end{equation} 
\end{definition}

Note that, using Lemma \ref{linearlemma}, a set of constraints of the form \eqref{exchangeableconstraint} can always be re-written as a single constraint (eliminating the dependence on $k$). However, for later purposes it is useful to refer to the general form \eqref{exchangeableconstraint}.

\section{Constrained quantum de Finetti theorem}
\label{maintheoremsection}

In order to give a unified proof of a constrained de Finetti theorem for processes and states, let us first note that Definition \ref{processexchangeability} of process exchangeability implies state exchangeability:
\begin{lemma}
If, for $n\geq 1$, $W^{(n)}$ define a sequence of exchangeable process matrices,  then the normalised operators $\rho^{(n)}:=W^{(n)}/\tr W^{(n)}$ define a sequence of exchangeable states, i.e., they satisfy the conditions
\begin{enumerate}
	\item \textbf{Symmetry}. $\rho^{(n)}$ is symmetric under permutations
	\item \textbf{State extendibility.} For every $n$, 
	\begin{equation}
	\rho^{(n)} = \tr_{n+1}\rho^{(n+1)}.
	\label{stateextendibility}
	\end{equation}
\end{enumerate}
\end{lemma}
To see that this holds, it is sufficient to note that $\widebar{M}=\id/d^O$ is the Choi representation of a CPTP map (the maximally depolarising channel). By substituting this map into the definition of process extendibility, Eq.~\eqref{pmextendibility}, we obtain state extendibility, Eq.~\eqref{stateextendibility}, after normalisation.
 
Note that the result above only works one way: even if we assume that $\rho^{(n)}$ satisfy the linear constraints for processes, state extendibility is a strictly weaker condition than process extendibility, as discussed in Appendix \ref{statevsprocess}. However, the mapping from exchangeable processes to exchangeable states is all we need to invoke the standard quantum de Finetti theorem:
\begin{theorem}[De Finetti theorem for quantum states \cite{Stoermer1969, Caves2002a}]
If, for $n\geq 1$, $\rho^{(n)}$ is a sequence of exchangeable states, then there is a unique probability measure $P$ over the space of single-trial density operators ${\mathfrak{S}}$ ($\rho\geq 0$, $\tr\rho=1$) such that
\begin{equation}
	\rho^{(n)}= \int_{\mathfrak{S}} d\rho P(\rho)   \rho^{\otimes n}.
	\label{productrepresentation}
\end{equation}
\end{theorem}

After restoring the normalisation factors, Eq.~\eqref{productrepresentation} \emph{almost} gives us a de Finetti representation for processes. The reason this is not quite the desired result is that the integral extends over the whole space of single-trial density operators, which also includes operators that are not valid processes. In order to interpret Eq.~\eqref{productrepresentation} as a mixture of i.i.d.\ processes, we have to show that the probability measure $P$ has support only on the space of valid single-trial processes, that is, processes that satisfy the constraint \eqref{normalisation}.

As discussed in the previous section, we prove a more general result that applies to all states subject to de Finetti-type constraints, Eq.~\eqref{exchangeableconstraint}. We want to show that a set of states $\rho^{(n)}$ satisfying constraints of the form \eqref{exchangeableconstraint} can be written as a mixture of i.i.d.\ states subject to the same constraint at the single-trial level. Perhaps surprisingly, a much weaker assumption leads to the desired result: it is sufficient to assume that the constraints hold for two trials. Our central result can then be formulated as follows:

\begin{theorem}\label{maintheorem}
For a set of indices $k$, given a set of real vector spaces $V^k$, a measurable space $Y$, a set linear functions $L^k_y:\mathfrak{H}(\mathcal{H}_\T)\rightarrow V^k$ for $y \in Y$, and a set of strictly positive measures $q^k$ over $Y$, a state $\rho^{(n)}$ that is 
\begin{enumerate}
	\item exchangeable,
	\item subject to a set of constraints of the form \eqref{exchangeableconstraint} for $n=2$: 
	\begin{equation}
\int_Y dy \, q^k(y) L^k_y\otimes L^k_y (\rho^{(2)}) =0\quad \forall k
\label{sumandproduct}
\end{equation}
\end{enumerate}
has a unique representation of the form
\begin{equation}
\rho^{(n)} = \int_{L^k_y(\rho)=0\,\mathrm{a.e.} y\in Y, \,\forall \,k }{ d\rho P(\rho) \rho^{\otimes n}}
\label{constraineddefinetti}
\end{equation}
for some probability measure $P(\rho)$, where the notation denotes an integral ranging over single-trial states subject to the constraints $L^k_y(\rho)=0$ $\forall\,k$ and a.e.\ for $y\in Y$.
\end{theorem}

Interestingly, this theorem can \emph{almost} be derived as the asymptotic limit of a result in Ref.~\cite{Berta2022}, Theorem 2.3, which considers a finite version of a constrained de Finetti theorem. However, the constraints considered in Ref.~\cite{Berta2022} are stronger than in our theorem, as they only apply to one factor:
$\mathcal{I}^{\otimes (n-1)} \otimes L \, (W^\n) = W^{(n-1)}\otimes v$ for some $v$ in the co-domain of $L$  and where $\mathcal{I}$ denotes the identity function. It was left as an open question if a weaker condition, as in our theorem, would still lead to a de Finetti theorem. Our result answers the question in the affirmative for the asymptotic limit of infinite exchangeable sequences.

Let us now prove Theorem \ref{maintheorem}.
\begin{proof}
For each vector space $V^k$, consider an arbitrary dual vector $\omega^k$, that is, an arbitrary linear function $\omega^k:V^k\rightarrow \mathbb{R}$. Applying $\omega^k\otimes\omega^k$ to Eq.~\eqref{sumandproduct}, we have, for every $k$, 
\begin{multline}
\int_Y dy \, q^k(y) \times \\
\times \left(\omega^k \circ L^k_y\right)\otimes \left(\omega^k \circ L^k_y\right) (\rho^{(2)}) =0.
\label{integralproduct}
\end{multline}
Since $\rho^{(2)}$ is exchangeable, we can expand it in the de Finetti form, Eq.~\eqref{productrepresentation}, which substituted into Eq.~\eqref{integralproduct} gives
\begin{equation}
 \!\!\!\int_Y \,dy \,q^k(y) \int_{\mathfrak{S}} d\rho P(\rho) \left(\omega^k \circ L^k_y(\rho) \right)^2 = 0.
\label{expanded}
\end{equation}
By assumption, $P(\rho)\geq 0$, $q^k(y)>0$ a.e., and $\omega^k \circ L^k_j(\rho)\in \mathbb{R}$, so the expression \eqref{expanded} is a positive linear combination of non-negative terms $\left(\omega^k \circ L^k_j(\rho) \right)^2$. This means that the expression can only vanish if
\begin{equation}
\begin{split} &P(\rho) \left(\omega^k \circ L^k_y(\rho) \right)^2 = 0 \\
&\forall k \text{ an a.e.\ for }  \rho \in \mathfrak{S},\,y\in Y.
\end{split}
\label{vanishing}
\end{equation}
This is only possible if, for each $k$ and for almost all $\rho$, $y$, at least one of the two factors vanishes, which means that either 
\begin{itemize}
	\item $P(\rho) = 0$  or
	\item $\omega^k \circ L^k_y(\rho) = 0$.
\end{itemize}
Since this has to hold for every dual vector $\omega^k$, we conclude that $P(\rho)=0$ a.e.\ for all $\rho$ such that $L^k_y(\rho) \neq 0$. Therefore, the integral in the de Finetti representation for $\rho^{(n)}$ can be limited to those states that satisfy $L^k_y(\rho) = 0$ $\forall k$ and a.e.\ for $y\in Y$. This yields the desired result, Eq.~\eqref{constraineddefinetti}.
\end{proof}

We will now discuss some consequences of this result.

\subsection*{General process matrices}
Because of Lemmas  \ref{unbiasedconstraintlemma} and \ref{linearlemma},  Theorem \ref{maintheorem} immediately implies that a sequence of exchangeable states subject to product expectation constraints, Eq.~\eqref{compatible}, for some set $\mathfrak{R}$ of operators  and some function $r$, has a unique de Finetti representation over states subject to the same constraints at the single-trial level:
\begin{equation}
\rho^{(n)} = \int_{\tr R \rho = r(R)\, \forall\, R \in \mathfrak{R}}{ d\rho P(\rho) \rho^{\otimes n}}.
\label{constraineddefinetti2}
\end{equation}
In particular, by choosing $\mathfrak{R}$ as in \eqref{CPTPchoi} (CPTP maps, up to normalisation), $r(R)=1$, and after the appropriate rescaling of operators, we obtain our original goal:
\begin{corollary}
An $n$-trial process matrix $W^{(n)}$ is exchangeable if and only if it has a de Finetti representation 
\begin{equation}
	W^{(n)}=\int_{\mathfrak{W}} dW P(W)  W ^{\otimes n},
	\label{processdeFinetti}
\end{equation}
for a unique probability measure $P(W)$, where $\mathfrak{W}$ denotes the set of single-trial process matrices.
\end{corollary}

\subsection*{Processes with no-signalling constraints}

Let us now consider $n$-trial processes subject to specific no-signalling constraints. These can arise if the relative spatiotemporal location of some sets of sites is fixed and known, or from structural properties of the experimental protocol, such as isolation between components of a device. For example, if $B^j$ is in the future of $A^j$ in every trial $j$, we can impose that $W^{(n)}$ permits no signalling from $B^j$ to $A^j$. The goal is to show that we can also restrict the single-trial probability $P$ to have support on processes that satisfy the same no signalling constraint, i.e., that we can interpret our scenario as consisting of i.i.d.\ processes with no signalling from $B$ to $A$.

As shown explicitly in Ref.~\cite{Wechs2019}, no-signalling from a set of parties $\mathcal{K}_2$ to a disjoint set of parties $\mathcal{K}_1$ can be characterised through a linear constraint
\begin{equation}
L^{\mathcal{K}_1\prec \mathcal{K}_2}(W)=0,
\label{multipartnosignalling}
\end{equation}
where we refer to Ref.~\cite{Wechs2019} for the explicit form of the linear function $L^{\mathcal{K}_1\prec \mathcal{K}_2}:\mathfrak{L}(\mathcal{H}_\T)\rightarrow \mathfrak{L}(\mathcal{H}_\T)$.  A given scenario may give rise to multiple no-signalling constraints, corresponding to a set of linear functions $L^k$, with each $k$ labelling no-signalling from a set of sites to another set of sites.

Given an $n$-trial process matrix $W^{(n)}$, imposing a set of no-signalling conditions for each trial results in the set of constraints
\begin{align*}
L^k\otimes\mathcal{I}\cdots&\otimes\mathcal{I} (W^\n) = 0 \quad \forall k, \\
\mathcal{I}\otimes L^k\otimes\cdots&\otimes\mathcal{I} (W^\n) = 0 \quad \forall k, \\
\dots&
\label{singlesiteconstraints}
\end{align*}
These constraints can be combined to give
\begin{equation}
L^k\otimes\cdots\otimes L^k (W^\n) = 0 \quad \forall k.
\label{combinednosignalling}
\end{equation}
This is now a set of de Finetti-type constraints, Eq.~\eqref{exchangeableconstraint}, so we can apply the constrained de Finetti theorem and obtain the following result:
\begin{corollary}
An exchangeable $n$-trial process matrix $W^\n$, subject to single-site no-signalling constraints encoded in a set of linear functions $L^k$, has a unique representation as a mixture of i.i.d.\ single-site process matrices subject to the same no-signalling constraints:
\begin{equation}
	 W^\n=\int_{\subalign{&W\in \mathfrak{W}\\ 
	&L^k(W)=0 \,\forall k}} 	\! dW P(W)  W^{\otimes n}.
	\label{nosignallingdeFinetti}
\end{equation}
\end{corollary}

In particular, this result applies to a typical experimental setup where, in each run of the experiment, a sequence of temporally separated measurements is performed. In such a case, we are entitled to impose one-way no-signalling constraints to the global $n$-trial process matrix $W^\n$. If we can also assume exchangeability, then we can write the $n$-trial process matrix as a mixture of i.i.d.\ causally ordered processes, recovering a \emph{de Finetti theorem for combs}, which was proven independently \cite{Puseyprivate}.

\section{Frequentist interpretation of the constraints}
\label{frequentist}

In line with the original de Finetti theorem, we have presented our result within  a Bayesian mind set, where probabilities, states, and processes are representations of an agent's knowledge of (or expectations about) a given scenario, without any a priori frequentist interpretation. The frequentist interpretation emerges for the single-trial probabilities appearing in the de Finetti representation.

However, it can be useful to re-interpret the result in a frequentist approach, as this can give some intuition about when the state constraints we have introduced may arise. For simplicity, we will only consider states, although the discussion directly extends to processes.

In a frequentist view, the $n$-trial state $\rho^\n$, and any probability we associate to it, describes a scenario where the entire set of trials can be repeated an arbitrary number of times. We refer to a set of $n$ trials as a \emph{supertrial}. It is of course possible to recover the frequentist interpretation of $\rho^\n$ by enlarging the Bayesian description to include an unbounded number of supertrials. The fact that each supertrial is represented by the same, `unknown' $\rho^\n$ is then recovered by assuming exchangeability of the sequence of supertrials. However, we will stick to a frequentist language for describing the repetitions of supertrials.

For concreteness, we can imagine a source of quantum systems that is turned on every day. Each day, the source produces $n$ systems, which we collectively describe by some state $\rho^\n$, and on which we can perform arbitrary measurements. After repeating the same measurements for an arbitrary number of days, the frequency of any measurement converges to the probability calculated from $\rho^\n$. If $\rho^\n$ is exchangeable (and hence has a de Finetti representation), it means that, within any given day, the source prepares the same state $\rho$  over and over, but a different $\rho$ is picked randomly each day according to the probability $P(\rho)$, so that the supertrial state is represented by the mixture $\rho^\n = \int_{\mathfrak{S}} P(\rho) \rho^\on$.

The frequentist language makes it easier to formulate a context in which constraints on the expectation values of observables, such as in Eq.~\eqref{compatible}, may arise. Imagine that, each day, we only use the source twice, and each time we measure some chosen observable $R$. This means that we are measuring the product observable $R\otimes R$ over the two-trial state $\rho^{(2)} = \tr_{3\dots n} \rho^\n$.  After repeating the experiment over many days, we record the one-trial and two-trial expectation values
\begin{align} \label{singletrialexpectation}
\langle R \rangle &= \tr\left( R \rho^{(1)} \right), \\
\label{twotrialexpectation}
\langle R\otimes R \rangle &= \tr\left( R\otimes R \rho^{(2)} \right).
\end{align}
Note that, at this point, we do not know what $\rho^\n$ is---we just assume that there must be some $\rho^\n$ that reproduces our observed expectation values according to expressions \eqref{singletrialexpectation} and \eqref{twotrialexpectation}, where $\rho^{(1)} = \tr_{2} \rho^{(2)}$.

In this scenario, we say that we have a product expectation constraint if
\begin{equation} \label{twiceconstrained}
\langle R\otimes R \rangle = \langle R \rangle^2.
\end{equation}
Indeed, this coincides with the form \eqref{compatible} where the set of observables $\mathfrak{R}$ contains a single element $R$, $r(R) = \langle R \rangle$, and $n=2$.

The constrained de Finetti theorem implies that, having observed the relation \eqref{twiceconstrained} (and assuming exchangeability), we can conclude that, at each trial, the source always produces a state $\rho$ with expectation value 
\begin{equation}
 \label{oneexpectation}
\tr \left( R \rho \right) = r(R) = \langle R \rangle.
\end{equation} 
Note that this is not the same as the expectation value \eqref{singletrialexpectation}, which is obtained from one measurement per day repeated over many days, for which we use $\rho^{(1)} = \int_{\mathfrak{S}}d\rho P(\rho) \rho$. Eq.~\eqref{oneexpectation} means that if, over a \emph{single} day, we measure $R$ arbitrarily many times (not just two), the average of all observed results will approach $\langle R \rangle$. More general product expectation constraints arise in analogous scenarios for a larger number of observables.

As a concrete example, if the single-trial system has two levels, and we take the Pauli observable $R=Z$, then the two-trial constraint $\langle Z \otimes Z \rangle = \langle Z \rangle^2 =: \bar{z}^2$  on an exchangeable state $\rho^{(n)}$ implies\footnote{A similarly constrained state was considered in Ref.~\cite{Schack2001}, although there the constraint emerged from arbitrarily many measurements of $Z$ within a single supertrial.}
\begin{multline}
\rho^{(n)} = \int_{x^2+y^2\leq 1-\bar{z}^2} dx dy \, P(x,y) \times \\
\times \frac{1}{2^n}\left(\id + x X + y Y + \bar{z} Z \right)^\on,
\end{multline}
where $X$ and $Y$ are the two remaining Pauli operators. This tells us that the source always prepares a state in the section of the Bloch sphere with $z$ coordinate equal to $\bar{z}$, but the $x$ and $y$ coordinates are picked randomly each day according to the distribution $P(x,y)$.

\section{(Non)-extensions of the result} 

It is natural to ask whether there are other types of constraints, not covered by Theorem \ref{maintheorem}, that still lead to a constrained de Finetti representation. It is in fact quite instructive that this does not work for some seemingly mild generalisations of our constraints. In rather general terms, we can formulate the question as follows:

\emph{Given a sequence of exchangeable states $\rho^\n$ and a sequence of functions $F^\n$ such that, for some values of $n$, $F^\n(\rho^\n)=0$, does $\rho^\n$ have a de Finetti representation with support on states satisfying $F^{(1)}(\rho)=0$? In other words, can we write
\begin{equation}
\rho^{(n)} = \int_{F^{(1)}(\rho)=0}{d\rho \, P(\rho) \rho^{\otimes n}}
\label{generalconstraineddeFinetti}
\end{equation}
for some probability measure $P(\rho)$?} We explore examples of this type in Appendix \ref{appendixconstraints}. Here we only summarise the key findings.

\begin{itemize}
	\item Instead of constraints that fix the expectation values of some observables through equalities, as in Eq.~\eqref{compatible}, we can consider \emph{inequalities} instead. As it turns out, inequality constraints such as
	$\tr\left[\left(\bigotimes_{j=1}^n  R_j \right)\rho^\n\right] \geq \prod_{j=1}^n r(R_j)$,  \emph{do not} imply that the de Finetti representation can be constrained to states with $\tr\left( R_j \rho \right)\geq r(R_j)$.
	\item In Theorem \ref{maintheorem}, we have seen that it is sufficient to impose the linear constraints on $\rho^{(2)}$ in order to derive a constrained de Finetti representation. By contrast, imposing the constraints only on $\rho^{(1)}$, or on $\rho^\n$ for any odd $n$, is generally not sufficient to derive a constrained de Finetti representation.
	\item However, if $R$ is positive semidefinite, it is sufficient to impose $\tr R \rho^{(1)} =  0$ to derive a de Finetti representation constrained on states with $\tr R \rho =  0$. Since $R>0$ can be interpreted as a measurement operator, the implication is that if a measurement outcome cannot occur in a single trial, then the `unknown state' must also give $0$ probability for that outcome.
	\item If, on top of exchangeability, we require $\rho^{(n)}$ to be invariant under the action of a group acting on the single-trial space, we \emph{cannot} conclude that the corresponding de Finetti representation is constrained to states invariant under the group's action.
\end{itemize}

\section{Conclusions}

At the technical level, we have shown that an exchangeable state or process subject to a product of linear constraints, or mixtures thereof, can always be expressed as a mixture of product states or processes, with each factor satisfying the same constraints. In particular, our result implies that exchangeable processes with arbitrary causal structure are mixtures of i.i.d.\ processes; that exchangeable multi-time, causally ordered processes are mixtures of i.i.d.\ multi-time, causally ordered processes; and that exchangeable processes with specified no-signalling constraints are mixtures of i.i.d.\ processes with the same no-signalling constraints.

Our result shows that an exchangeability assumption can ground the notion of repeatability for experiments involving quantum causal structure. The result clarifies under what conditions we are entitled to regard an experiment as multiple repetitions, \emph{under equal conditions}, of a single experiment involving multiple events. This paradigm entitles us to `discover an unknown causal structure' without any ontological commitment regarding quantum states, processes, or causal structure: as long as we can justify the equivalence of different trials under permutations, and the possibility to repeat an experiment indefinitely, we can treat our scenario \emph{as if} it is governed by a `real' underlying process matrix $W$, which in turn encodes all potentially accessible information about causal relations.

In practice, once exchangeability is established, the distribution $P(W)$ appearing in the de Finetti representation, Eq.~\eqref{processdeFinetti}, effectively describes prior knowledge of the process of interest. Following a similar argument as for states \cite{Schack2001}, observing a set of outcomes $\vec{a}$ in each trial leads to an update of the prior according to Bayes rule:
\begin{equation}
P_{\textup{upd}}(W|\vec{a}) = P(W) \frac{P(\vec{a}|W)}{P(\vec{a})},
\label{processbayes}
\end{equation}
where $P(\vec{a}|W)$ is calculated using the Born rule for processes, Eq.~\eqref{processborn_Npartite}, and $P(\vec{a})$ is the marginal outcome probability given the prior: $P(\vec{a})= \int dW P(\vec{a}|W) P(W)$. 
Crucially, given prior information about the causal relations between the events involved, we can constrain accordingly the prior $P(W)$, without additional assumptions apart from exchangeability.
This opens the task to formalise the procedure into concrete algorithmic routines, extending existing methods for quantum states and quantum channels \cite{Granade2017qinferstatistical}.

	Just like the original de Finetti theorem, our result relies on the unphysical assumption that the experiment can be repeated an arbitrarily large number of times. However, finite de Finetti theorems exist for classical probabilities \cite{Diaconis1977, Diaconis1980}, quantum states \cite{Koenig2005, Koenig2009}, and quantum channels \cite{Berta2022}. The general idea is that a subsystem of a finite, symmetric state (or channel, process, etc.), should approximate a state (or channel, process, etc.) in the de Finetti form. Formulating a finite version of a de Finetti theorem for processes could provide a concrete starting point to treat practical scenarios where full exchangeability cannot be guaranteed. As noted earlier, finding a finite version of a constrained de Finetti theorem such as Theorem \ref{maintheorem} was also posed as an open question in Ref.~\cite{Berta2022}.

Finally, it is interesting that the constraints in our theorem do not include causal separability \cite{araujo15}. This is relevant, for example, in a scenario where all operations are well localised in time, but there is uncertainty about their order: In this case, one can constrain processes to probabilistic mixture of causally ordered ones, which is a special instance of causal separability \cite{oreshkov15, Wechs2019}. However, working from first principles, this only holds for the entire, one-shot process. Even assuming exchangeability, our current result \emph{does not} guarantee that a causally separable process can be written as a mixture of i.i.d.\ causally separable ones, and it is a compelling open question whether such an extension holds.

\begin{acknowledgments}
We thank Eric Cavalcanti, Andrew Doherty, Omar Fawzi, Matthew Pusey, and Frank Wilczek for valuable comments and discussions. This work has been supported by the Australian Research Council (ARC) Centre of Excellence for Engineered Quantum Systems (EQUS, CE170100009) and by the John Templeton Foundation through the ID \# 62312 grant, as part of the ‘The Quantum Information Structure of Spacetime’ Project (QISS). The opinions expressed in this publication are those of the authors and do not necessarily reflect the views of the John Templeton Foundation. Nordita is supported in part by NordForsk. We acknowledge the traditional owners of the land on which the University of Queensland is situated, the Turrbal and Jagera people.
\end{acknowledgments}



\appendix
\section{On the definition of extendibility}
\label{otherextendibility}

The definition of process extendibility, Eq.~\eqref{pmextendibility}, is rather strong: we ask that we recover the same $W^{(n)}$ from $W^{(n+1)}$ for every CPTP map $\widebar{M}^{\,n+1}$ (in particular, together with symmetry, this implies no signalling across different trials). This is a strictly stronger condition than state extendibility, for which we only ask to recover $W^{(n)}$ for $\widebar{M}^{\,n+1} = \id/d^O$---we show this with an explicit example in Sec.~\ref{statevsprocess}. As we have seen, state extendibility is sufficient to prove our main result, Theorem \ref{maintheorem}, but we can ask whether other definitions work too. We look at two meaningful options, one turning out to be too weak, while the second still being sufficient to prove the de Finetti theorem for processes.

\subsection{Weak extendibility}\label{weak}
It is meaningful to consider scenarios where we only impose that the $n$-trial process is recovered from the $n+1$-trial one for some particular CPTP map.
\begin{definition}
A sequence of process matrices $W^{(n)}$, $n\geq 1$, is \emph{weakly extendible} if there is a particular CPTP map $\widebar{M}$ such that
	\begin{equation}
	W^{(n)} = \tr_{n+1}\left[W^{(n+1)}(\id^{A^1\dots A^n} \otimes \widebar{M})\right].
	\label{weakextendibility}
	\end{equation}
\end{definition}
We say that $W^{(n)}$ is \emph{weakly exchangeable} if it is symmetric and weakly extendible. The following counterexample shows that weak exchangeability is not sufficient to ensure a de Finetti representation for the processes $W^{(n)}$.

\textbf{Counterexample.} Consider a scenario with one site $A$ per trial, with isomorphic input and output spaces $A_I$, $A_O$ of finite dimension $d$. For $n$ trials, define the ``idle'' process
\begin{equation}
W^{A^1\rightarrow \cdots \rightarrow A^n} := \rho \otimes [[\id_d]]^\on\otimes \id_d,
\label{idlecomb}
\end{equation}
where $\id_d$ is the $d$-dimensional identity operator and we use the notation
\begin{equation}
 [[\id_d]] := \sum_{j,k=1}^d\ketbra{j}{k}\otimes \ketbra{j}{k}
\end{equation}
to denote the Choi representation of the identity channel (and, as in the rest of this work, we order tensor factors according to $A^1,\dots,A^n$).
The idle process is a causally ordered process where the first site, $A^1$, receives the state $\rho$, and each site is linked to the next one through the identity map, in the order $A^1\prec\dots \prec A^n$.

Now consider the symmetrised process
\begin{equation} \label{symmetrised}
W^{(n)}_{\textrm{sym}}:= \frac{1}{n!} \sum_{\sigma\in S^n} W^{A^{\sigma(1)}\rightarrow \cdots \rightarrow A^{\sigma(n)}},
\end{equation}
where the sum is over all $n$-element permutations. By construction, $W^{(n)}_{\textrm{sym}}$ is symmetric under permutations. We can also see that it satisfies weak extendibility for the CPTP map $\widebar{M} = [[\id_d]]$. Indeed, plugging the identity map into site $A^{n+1}$ of the permuted idle process, $W^{A^{\sigma(1)}\rightarrow \cdots \rightarrow A^{\sigma(n+1)}}$, is equivalent to  plugging the identity map into site $A^{\sigma(n+1)}$ of the unpermuted  process, Eq.~\eqref{idlecomb}, which, for any permutation $\sigma$, always gives back the idle process with one less site. Summing over all permutations then yields $\tr_{n+1}\left[ W^{(n+1)}_{\textrm{sym}}(\id^{A^1\dots A^n} \otimes  [[\id_d]])\right] = W^{(n)}_{\textrm{sym}}$, in agreement with the definition of weak extendibility. In summary, $W^{(n)}_{\textrm{sym}}$ is a weakly exchangeable process for every $n$, but it is clearly not in the de Finetti form\footnote{A simple way to prove this is to note that $W^{(n)}_{\textrm{sym}}$ allows some signalling between any pair of sites, which is not possible for processes in the de Finetti form.}, so weak exchangeability is not sufficient to deduce a de Finetti representation.

Note that weak exchangeability for $\widebar{M} = \id^{A_IA_O}/d^{A_O}$ reduces to state exchangeability, which we have seen does lead to a de Finetti theorem. Hence, although weak exchangeability does not work in general, it does for specific choices of CPTP $\widebar{M}$ in Eq.~\eqref{weakextendibility}.

\subsection{Channel extendibility}
A process can be seen as a channel from all output spaces to all input spaces, where the Choi representation of the channel coincides with the process matrix. Operationally, this corresponds to restricting each party's operations to be a measurement followed by an independent state preparation. The state that the parties receive and measure in their input spaces is given by the channel acting on the output spaces. In formulas, a process matrix $W \in \mathfrak{L}(\mathcal{H}^{A}\otimes\mathcal{H}^{B}\otimes \cdots)$ defines a channel $\mathcal{W}: \mathcal{L}(\mathcal{H}^{A_O}\otimes \mathcal{H}^{B_O}\otimes \cdots) \to \mathcal{L}(\mathcal{H}^{A_I}\otimes \mathcal{H}^{B_I}\otimes \cdots)$ acting as
\begin{multline}
      \mathcal{W}(\rho^{A_O B_O \dots}) \\
			:= \tr_{A_OB_O\dots} \left[\left(\rho^{A_O B_O \dots}\right)^T  W  \right].
\end{multline}

In this view, it is meaningful to consider a version of extendibility that follows the definition for channels given by Fuchs, Schack, and Scudo \cite{Fuchs2004}:
\begin{definition}
A sequence of process matrices $W^{(n)}$, $n\geq 1$, is \emph{channel extendible} if, for every state $\rho$,
	\begin{multline}
	W^{(n)} \\
	\!\!\!\!\!\!= \tr_{n+1}\left[W^{(n+1)}(\id^{A^1\dots A^n} \otimes \id^{A^{n+1}_I} \otimes \! \rho^{A^{n+1}_O})\right].
	\label{channelextendibility}
	\end{multline}
\end{definition}
We say that $W^{(n)}$ is \emph{channel exchangeable} if it is symmetric and channel extendible. 

According to this definition, we only require to recover $W^{(n)}$ when we restrict $A^{n+1}$ to ``trace and re-prepare'' operations, where the input is traced out and an arbitrary state is prepared. 
For example, the process $W^{A^1\rightarrow \cdots \rightarrow A^n}$ defined in Eq.~\eqref{idlecomb} satisfies channel extendibility (because any state prepared at the last site simply gets traced out), although it does not satisfy symmetry. On the other hand,  the symmetrised process $W^{(n)}_{\textrm{sym}}$, Eq.~\eqref{symmetrised} is not channel extendible.

It is easy to see that channel extendibility can replace process extendibility in the de Finetti theorem. Indeed, fixing the re-prepared state $\rho$ to be the maximally mixed one, channel extendibility implies that the normalised states $W^\n/\tr W^\n$ are state extendible, which is sufficient to prove the constrained de Finetti theorem.

\subsection{State extendibility does not imply process extendibility}
\label{statevsprocess}
We have mentioned in Sec.~\ref{maintheoremsection} of the main text that state extendibility does not imply process extendibility. Let us see an explicit example.

\textbf{Example.} Consider a scenario with one site per trial, with isomorphic input and output spaces of dimension $d$. Consider a particular case of the idle process, Eq.~\eqref{idlecomb}, with initial maximally mixed state $\rho= \id_d/d =: \id^{\circ}$, but with reversed causal order $A^n\prec \dots \prec A^1 $:
\begin{multline}
W^{(n)} = W^{A^n\rightarrow \cdots \rightarrow A^1}_{\id^{\circ}} \\
 := U(r_n) \id^{\circ} \otimes [[\id_d]]^\on\otimes \id_dU^{\dag}(r_n),
\label{idleswapped}
\end{multline}
where $r_n$ is the reverse permutation of $n$ elements, $r_n(j) := n+1 - j$, and $U(r_n)$ its unitary representation as defined in Eq.~\eqref{permutations}. 

To see that $W^{A^n\rightarrow \cdots \rightarrow A^1}_{\id^{\circ}}$ is not process extendible (in fact, that is not even channel extendible), consider a trace-replace CPTP map $\widebar{M} = \id_d \otimes \rho$, where $\rho\neq\id^{\circ}$ is an arbitrary non-maximally-mixed state. Plugging $\widebar{M}$ into the site $A^{n+1}$ (which is the first in the causal order) of $W^{(n+1)}$ gives an $n$-site idle process, in the order $A^n\prec\dots \prec A^1$, starting with $\rho$:
\begin{multline}
\tr_{n+1}\left[W^{(n+1)}(\id^{A^1\dots A^n} \otimes \widebar{M})\right] \\
\tr_{n+1}\left[W^{A^n\rightarrow \cdots \rightarrow A^1}_{\id^{\circ}} (\id^{A^1\dots A^n} \otimes \id_d \otimes \rho )\right] \\
= U(r_n) \rho \otimes [[\id_d]]^\on\otimes \id_dU^{\dag}(r_n) \\
=: W^{A^n\rightarrow \cdots \rightarrow A^1}_{\rho},
\end{multline}
which is not the same as $W^{(n)}$ as defined in Eq.~\eqref{idleswapped}.

On the other hand, substituting  $\rho = \id^{\circ}$ in the above calculation gives back $W^{A^n\rightarrow \cdots \rightarrow A^1}_{\id^{\circ}} = W^\n$. This implies that the sequence of states $\rho^\n := W^\n/d^{n}$ satisfies state extendibility. Hence we have found a sequence of (causally ordered) processes that is state extendible (after normalisation) but not process (nor channel) extendible.

\section{Proofs of lemmas}
\label{lemmaproofs}

\subsection{Proof of Lemma \ref{unbiasedconstraintlemma}}
To simplify the notation, let us write $r_j:= r(R_j)$, so we can restate the lemma as
\begin{lemma*}
A sequence of exchangeable states $\rho^{(n)}$ satisfies product expectation constraints, of the form 
\begin{equation}
\tr\left[\left(\bigotimes_{j=1}^n  R_{j} \right)\rho^{(n)}\right] = \prod_{j=1}^n r_j
	\label{Acompatible}
\end{equation}
for a set of operator $\mathfrak{R}\subset \mathfrak{L}\left(\mathcal{H}_\T\right)$, if and only if it satisfies
\begin{equation}
\tr\left[\left(\bigotimes_{j=1}^n  \sigma_{j} \right)  \rho^{(n)}\right] = 0
\label{linearcompatibleA}
\end{equation}
for all 
\begin{equation}
\sigma_j \coloneqq R_j - r_j, 
\label{RtosigmaA}
\end{equation} 
$R_j\in\mathfrak{R}$, $j=1,\dots,n$.
\end{lemma*}

\begin{proof}

Let us introduce the notation
\begin{equation}
X^{0}_j := R_j , \qquad
X^{1}_j := r_j
\end{equation} 
(where, as usual, identity operators are implied in $r_j\equiv r_j \id$).

We can expand the left hand side of Eq.~\eqref{linearcompatibleA} as
\begin{multline} \label{expansion}
\tr\left[\left(\bigotimes_{j=1}^{n}\sigma_j \right)\rho^{(n)}\right] \\
 \!\! = \!\!\! \sum_{\mu_1\dots \mu_n=0,1}\!\!\!  \tr\left[\left(\bigotimes_{j=1}^{n}(-1)^{\mu_j}X_j^{\mu_j} \right)\rho^{(n)}\right].
\end{multline}
Let us now assume that Eq.~\eqref{Acompatible} holds for all $n$. Using the exchangeability of $\rho^\n$, and hence its de Finetti representation, we can see that 
\begin{equation}
\tr\left[\bigotimes_{j=1}^{n}X_j^{\mu_j} \rho^{(n)}\right] = \prod_{j=1}^n r_j
\end{equation}
for all values of $\mu_1,\dots,\mu_n=0,1$. So Eq.~\eqref{expansion} becomes
\begin{multline*}
\sum_{\mu_1\dots \mu_n=0,1}\!\!\!  \tr\left[\left(\bigotimes_{j=1}^{n}(-1)^{\mu_j}X_j^{\mu_j} \right)\rho^{(n)}\right] \\
= \prod_{j=1}^n\left(\sum_{\mu_j=0,1} (-1)^{\mu_j} r_j\right)  = 0.
\end{multline*}
This proves that condition \eqref{Acompatible} implies \eqref{linearcompatibleA}.

For the other direction, we can proceed by induction. 
For $n=1$, using $\tr\rho^{(1)} = 1$ we see directly that 
\begin{equation} 
\tr \left[(R_1 -r_1) \rho^{(1)}\right] = 0
\end{equation}
is equivalent to 
\begin{equation} 
\tr \left[R_1 \rho^{(1)}\right] = r_1.
\end{equation}
Assume now that Eq.~\eqref{linearcompatibleA} holds for some $n=\bar{n}$ and that Eq.~\eqref{Acompatible} holds for all $n<\bar{n}$. This implies 
\begin{equation}
\tr\left[\bigotimes_{j=1}^{\bar{n}}X_j^{\mu_j} \rho^{(\bar{n})}\right] = \prod_{j=1}^{\bar{n}} r_j
\end{equation}
for $\sum_{j=1}^{\bar{n}}\mu_j >0$; that is, for all choices of $\mu_1,\dots,\mu_{\bar{n}}$ except $\mu_1 = \dots =\mu_{\bar{n}} = 0$, which corresponds to the term $\tr\left[\bigotimes_{j=1}^{\bar{n}}R_j \rho^{(\bar{n})}\right]$. 
Using again the  expansion \eqref{expansion}, we have 
\begin{multline}
\tr\left[\left(\bigotimes_{j=1}^{\bar{n}}\sigma_j \right)\rho^{(n)}\right] \\
=\tr\left[\bigotimes_{j=1}^{\bar{n}}R_j \rho^{(\bar{n})}\right]  - \prod_{j=1}^{\bar{n}} r_j.
\end{multline}
Together with \eqref{linearcompatibleA}, this implies \eqref{Acompatible} for an arbitrary $\bar{n}$, concluding the proof.
\end{proof}

\subsection{Proof of Lemma \ref{linearlemma}}
\label{linearlemmaappendix}
Let us re-state the lemma:
\begin{lemma*}
Given two inner product spaces $V^1$, $V^2$, a measurable space $Y$, a set of linear functions $L_y: V^1\rightarrow V^2$, $y\in Y$, and for any $v\in V^1$
\begin{equation} 
L_y(v)=0 \quad \text{\emph{a.e. for }} y\in Y
\end{equation}
if and only if, for any strictly positive measure $q$ over $Y$ ($q(y)>0$ a.e.)
\begin{equation}
\int_Y dy \, q(y)L^{\dag}_y\circ L_y(v)=0,
\end{equation}
\end{lemma*}

\begin{proof}
If $L^{\dag}_y = 0$ a.e.\ for $y \in Y$, then it is clear that 
\begin{equation} \label{integraly}
\int_Y dy \, q(y)L^{\dag}_y\circ L_y(v)=0
\end{equation}
for any measure $q$. We need to prove that, for any strictly positive $q$, Eq.~\eqref{integraly} implies $L_y(v)=0$ a.e.\ for $y \in Y$. To this end, it is sufficient to take the inner product of Eq.~\eqref{integraly} with $v$:
\begin{multline} \label{integralLy}
0=\braket{v}{\int_Y dy \, q(y)L^{\dag}_y\circ L_y(v)} \\
= \int_Y dy \, q(y) \braket{L_y(v)}{L_y(v)}.
\end{multline}
As the integrand is non-negative a.e., and $q(y)$ is strictly positive, Eq.~\eqref{integralLy} implies that $L_y(v) = 0$ a.e.\ in $Y$, which proves the lemma.
\end{proof}

\section{Other classes of constraints}
\label{appendixconstraints}
Here we consider some modifications to the hypothesis of Theorem \ref{maintheorem}, some of which lead to a corresponding de Finetti representation and some of which do not. Below, it will always be assumed that $\rho^{(n)}$ are exchangeable states for all $n$.

\subsection{Inequality constraints}
\label{inequalityconstraints}
As in the case of product expectation constraints, consider a set of single-trial observables $\mathfrak{R}$ and a function $r:\mathfrak{R} \rightarrow \mathbb{R}$. However, assume we are only given a bound on the expectation values, in the form
\begin{equation}
tr\left[\left(\bigotimes_{j=1}^n  R_j \right)\rho^\n\right] \geq \prod_{j=1}^n r(R_j).
\label{ineqconstraints}
\end{equation}
for an exchangeable $\rho^{(n)}$. This constraint \emph{does not} imply a representation of the form
\begin{equation}
	\rho^\n =\int_{\tr\left( R \rho \right)\geq r(R)\, \forall R\in \mathfrak{R}} d\rho P(\rho)  \rho^{\otimes n}.
	\label{inequalitydefinetti}
\end{equation}

\textbf{Counterexample}
As single-trial space, let us take a two-dimensional system with basis vectors $\ket{0}$ and $\ket{1}$.  Consider an operator $R$ such that
\begin{equation}
\!\! \bra{0}R\ket{0}=r_0\geq 0, \quad \bra{1}R\ket{1}=r_1 > r_0.
\label{Hexpectations}
\end{equation}
Now consider the exchangeable states
\begin{equation}
\rho^{(n)}:= \frac{1}{2}\left(\proj{0}^{\otimes n} + \proj{1}^{\otimes n}\right).
\label{counterexample}
\end{equation}
These states satisfy
\begin{align}
\label{inequalityE}
\tr\left(R^{\on}\rho^{(n)}\right) =& \frac{1}{2}\left(r_0^{n} + r_1^{n}\right) \geq {r}^n, \\ \nonumber
  r \coloneqq & \frac{r_0+r_1}{2}.
\end{align}
Eq.~\eqref{inequalityE} is an inequality constraint of the form \eqref{ineqconstraints}. However, $\rho^{(n)}$ cannot be decomposed in the form \eqref{inequalitydefinetti}: Eq.~\eqref{counterexample} is already in the de Finetti form, where the probability
$$P(\rho) = \frac{1}{2}\left[\delta(\rho - \proj{0}) + \delta(\rho - \proj{1}) \right]$$
does not vanish on $\rho=\proj{0}$, and  $\tr\left( R \proj{0}\right) = r_0 < r$.

\subsection{Single- and odd-trial constraints} \label{singlecopy}
\subsubsection{Single- and odd-trial constraints without a constrained de Finetti representation}
Let us consider some single-trial operator $R$ and a real number $r$. We can find exchangeable states $\rho^{(n)}$ that satisfy the single-trial constraint $\tr R \rho^{(1)} = r$ but whose de Finetti representation does not have support limited to states with $\tr R \rho = r $. 

As a counterexample, we can take again the state in Eq.~\eqref{counterexample}, $\rho^{(n)}:= \frac{1}{2}\left(\proj{0}^{\otimes n} + \proj{1}^{\otimes n}\right)$, and as observable the Pauli operator $R=Z$. Since $\rho^{(1)}=\frac{\id}{2}$, we have $\langle Z\rangle = 0$, even though $\tr Z\rho = \pm 1$ for the two states appearing in the de Finetti representation, $\rho = \proj{0}, \proj{1}$. In fact, for any odd $n$ we have $\langle Z^\on\rangle = 0 = \langle Z \rangle^n$, so, in general, imposing a product expectation constraint only on an odd number of trials does not imply a constrained de Finetti representation.

Note also that, since $\rho^\n$ is symmetric, single-trial constraints can be expressed equivalently as $\frac{1}{n}\langle R\otimes \id \dots \otimes \id + \dots + \id \dots \otimes \id \otimes R \rangle = r$, namely as a constraint on the sample mean of the operator $R$. The failure of the constrained de Finetti representation tells us (the known fact) that fixing the expectation value of the sample mean of an observable does not fix the expectation value of that observable for the `unknown state'.

\subsubsection{Single-trial constraints with a constrained de Finetti representation}

Consider a self-adjoint operator $R\geq r$; that is, such that $R-r$ is positive semidefinite for some given real number $r$. In this case, given an exchangeable sequence of states $\rho^\n$, it is sufficient to impose the single-trial constraint
\begin{equation}
\label{singlecopyconstraint}
\tr R \rho^{(1)} = r
\end{equation}
to deduce that the de Finetti representation can be restricted to states $\rho$ such that $\tr R \rho = r$. 

\begin{proof}
 $R \geq r$ implies that, for every state $\rho$, 
$\tr R \rho \geq r$. Therefore, the de Finetti representation for a single trial gives 
\begin{equation} \label{expandedsingleexpectation}
\tr R \rho^{(1)} = \int_{\mathfrak{S}} d\rho P(\rho ) \tr R \rho \geq r,
\end{equation}
with equality only possible if $P(\rho) = 0$ a.e.\ for $\tr R \rho > r$. This means that we can write the $n$-trial state as
\begin{equation} 
\rho^{(n)} = \int_{\subalign{& \rho \in \mathfrak{S} \\ 
	& \tr R\rho \geq r}}  d\rho P(\rho ) \rho^\on.
\end{equation}
\end{proof}

\subsection{States invariant under joint group action -- constraints with the wrong sign} \label{groupinvariant}
Given a group $G$ and a unitary representation $U_g$ on the single-trial space, $g\in G$, consider exchangeable states $\rho^{(n)}$ that are invariant under the joint action of $G$:
\begin{equation} \label{Uinvariant}
U_g^{\otimes n} \rho^{(n)} \left(U_g^{\dag}\right)^{\otimes n} = \rho^{(n)} \quad \forall g\in G.
\end{equation} 
One may ask whether states of this type can be decomposed as 
\begin{equation} \label{wernersymmetric}
\rho^{(n)} = \int_{U_g\rho U_g^{\dag} = \rho} d\rho \, P(\rho) \rho^{\otimes n},
\end{equation}
but this is not necessarily true. A counterexample is given by  
\begin{equation}
\rho^{(n)} = \int d\psi \proj{\psi}^{\otimes n},
\end{equation}
where $d\psi$ is the measure on pure states induced by the Haar measure 
 (i.e., the unique measure invariant under arbitrary unitary transformations).
 $\rho^{(n)}$ is exchangeable and invariant under joint local unitaries, however, the probability measure has support over pure states, which are not invariant under arbitrary unitaries.

It is interesting that constraint \eqref{Uinvariant} is \emph{almost} of the form \eqref{exchangeableconstraint} of a de Finetti-type constraint, but it has the wrong sign: defining the linear operators $L^g_1(\rho) := U_g\rho U_g^{\dag}$ and $L^g_2(\rho) := -\rho$, we can write the single-trial constraint as $(L^g_1 + L^g_2)(\rho) = U_g\rho^{(1)} U_g^{\dag} - \rho^{(1)} =0$. However, this choice of functions implies the wrong two-trial constraint: in
$$\sum_{j=1}^2 L^g_j\otimes L^g_j (\rho^{(2)}) = U_g\otimes U_g\rho^{(2)} U_g^{\dag}\otimes U_g^{\dag} + \rho^{(2)} =0$$
the sign in front of $\rho^{(2)}$ is the opposite as compared to \eqref{Uinvariant} for $n=2$.

\end{document}